\documentclass[prl,floatfix,showpacs,aps,letterpaper,twocolumn]{revtex4}

\usepackage{amsmath}
\usepackage{amsfonts}
\usepackage{amssymb}
\usepackage{amsthm}
\usepackage{graphicx}
\usepackage{color}

\newcommand{\ZZ}{{\mathbb Z}} 
\newcommand{\FF}{{\mathbb F}}
\newcommand{\PauliGr}{{\cal P}}
\newcommand{\ZGr}{{\cal Z}}
\newcommand{\XGr}{{\cal X}}
\newcommand{\CXGr}{{c{\cal X}}}
\newcommand{\CZGr}{{c{\cal Z}}}
\newcommand{\ZZtn}{{\ZZ_2^n}}
\newcommand{\supp}[1]{{\mathrm{supp}(#1)}}
\newcommand{\ZOn}{{\{0,1\}^n}}

\newtheorem{theorem}{Theorem}
\newtheorem{lemma}{Lemma}
\newtheorem{corollary}{Corollary}

\begin{document}

\title{Scalable randomized benchmarking of non-Clifford gates}
\author{Andrew W. Cross}
\email{awcross@us.ibm.com}
\affiliation{IBM T.J. Watson Research Center, 1101 Kitchawan Road, Yorktown Heights, New York 10598}
\author{Easwar Magesan}
\affiliation{IBM T.J. Watson Research Center, 1101 Kitchawan Road, Yorktown Heights, New York 10598}
\author{Lev S. Bishop}
\affiliation{IBM T.J. Watson Research Center, 1101 Kitchawan Road, Yorktown Heights, New York 10598}
\author{John A. Smolin}
\affiliation{IBM T.J. Watson Research Center, 1101 Kitchawan Road, Yorktown Heights, New York 10598}
\author{Jay M. Gambetta}
\affiliation{IBM T.J. Watson Research Center, 1101 Kitchawan Road, Yorktown Heights, New York 10598}

\date{8 October 2015}


\begin{abstract}
Randomized benchmarking is a widely used experimental technique to characterize the average error of quantum operations. Benchmarking procedures that scale to enable characterization of $n$-qubit circuits rely on efficient procedures for manipulating those circuits and, as such, have been limited to subgroups of the Clifford group. However, universal quantum computers require additional, non-Clifford gates to approximate arbitrary unitary transformations. We define a scalable randomized benchmarking procedure over $n$-qubit unitary matrices that correspond to protected non-Clifford gates for a class of stabilizer codes. We present efficient methods for representing and composing group elements, sampling them uniformly, and synthesizing corresponding $\mathrm{poly}(n)$-sized circuits. The procedure provides experimental access to two independent parameters that together characterize the average gate fidelity of a group element.
\end{abstract}

\maketitle

A key step to realizing a large-scale universal quantum computer is demonstrating that decoherence and other realistic imperfections are small enough to  be overcome by fault-tolerant quantum computing protocols \cite{gottesman09,raussendorf07}. Randomized benchmarking (RB)~\cite{emerson05,knill08,dankert09,magesan11} has become a standard experimental technique for characterizing the average error of quantum gates due in part to its insensitivity to state preparation and measurement errors. Benchmarking provides robust estimates of average gate fidelity \cite{magesan11,flammia14} and can characterize specific interleaved gate errors \cite{magesan12,gaebler12}, addressability errors \cite{gambetta12}, and leakage errors \cite{epstein14,wallman15,chasseur15}.

RB techniques that efficiently scale to many qubits have been limited to subgroups of gates in the Clifford group, since computations with this group are tractable \cite{magesan11}. However, the Clifford group is not enough for general quantum computations~\cite{gottesman97}.
 Previous work generalizes RB to groups that include non-Clifford gates \cite{barends14,dugas15}, but only on single qubits, a significant limitation.

In this Letter we present a scalable RB procedure that includes important non-Clifford circuits, such as circuits composed from $T=\sqrt[4]{Z}$ and controlled-NOT (CNOT) gates that naturally occur in fault-tolerant quantum computations. The $n$-qubit matrix groups we study are a generalization of the standard dihedral group and coincide in some cases with protected gates in stabilizer codes, such as $k$-dimensional color codes \cite{bombin07}. Circuits built from these gates cannot be universal but do constitute significant portions of magic state distillation protocols \cite{bravyi05,duclos15}, repeat-until-success circuits \cite{paetznick14}, and the vital quantum Fourier transform \cite{nc00}. We show that there are efficient methods for representing and composing group elements, sampling them uniformly, and synthesizing corresponding circuits whose size grows polynomially with the number of qubits $n$. The benchmarking procedure provides experimental access to two independent noise parameters through exponential decays of average sequence fidelities.

The quantum circuits we consider are products of CNOT gates $\Lambda_{12}(X)|u,v\rangle:=|u,u\oplus v\rangle$, bit-flip gates $X|u\rangle:=|u\oplus 1\rangle$, and single qubit $m$-phase gates $Z_m|u\rangle:=\omega_m^u|u\rangle$ where $\omega_m=e^{i2\pi/m}$. More concisely, the circuits of interest are given by the group
\begin{equation}
G_m := \langle \Lambda_{ij}(X), X(j), Z_m(j)\rangle/\langle \omega_m\rangle.
\end{equation}
We call this group a {\em CNOT-dihedral group} since it is generated by CNOTs and a single qubit dihedral group \footnote{A natural generalization of this group is $G_{p,m}:=\langle \Lambda^{(p)}(X), X(j), Z_m(j)\rangle/\langle \omega_m\rangle$ where $\Lambda^{(p)}(X)$ is a $p$-controlled-NOT gate.}. Although we prove certain results for general $m$, we focus mainly on the case of $m=2^k$ as this affords efficient benchmarking and contains various non-Clifford gates of interest, such as $T$, controlled-$S$, and controlled-controlled-$Z$ \footnote{Controlled-$S$ is defined by $\Lambda_{12}(S)|u,v\rangle:=i^{uv}|u,v\rangle$, and controlled-controlled-Z gate is defined by $\Lambda_{123}(Z)|u,v,w\rangle:=(-1)^{uvw}|u,v,w\rangle$ and is locally equivalent to a Toffoli gate.}.

{\em RB over $G_{2^k}$ --} The benchmarking procedure we present here both generalizes \cite{dugas15} \footnote{Reference \cite{dugas15} defines another $n$-qubit group $\langle \Lambda_{ij}(Z), X(j), Z_m(j)\rangle$. Our arguments imply that this group twirl produces a map with exponentially many parameters.} and extends naturally to interleaving gates to estimate individual gate fidelities~\cite{magesan12,gaebler12}. The procedure is as follows. Choose a sequence of $\ell+1$ unitary gates where the first $\ell$ gates are uniformly random elements $g_{j_1}$, $g_{j_2}$, $\dots$, $g_{j_\ell}$ of $G_{2^k}$ and the $(\ell+1)^{\mathrm{th}}$ gate is $g_{{\bf j}_\ell}^{-1}:=g_{j_1}^\dagger\dots g_{j_\ell}^{\dagger}$ where ${\bf j}_\ell$ denotes the $\ell$-tuple $(j_1,\dots,j_\ell)$ labeling the sequence. We show later that elements of $G_{2^k}$ can be efficiently sampled and $g_{{\bf j}_\ell}^{-1}$ can be efficiently computed. For each sequence, we prepare an input state $\rho$, apply $S_{{\bf j}_\ell}:=g_{{\bf j}_\ell}^{-1}g_{j_\ell}\dots g_{j_1}$, and measure an operator $E$. The overlap with $E$ is $\mathrm{Tr}[ES_{{\bf j}_\ell}(\rho)]$. Averaging this overlap over $K$ independent sequences of length $\ell$ gives an estimate of the average sequence fidelity $F_{\mathrm seq}(\ell,E,\rho):=\mathrm{Tr}[ES_\ell(\rho)]$ where $S_\ell(\rho):=\frac{1}{K}\sum_{{\bf j}_\ell} S_{{\bf j}_\ell}(\rho)$ is the average quantum channel.

The error operator of the final gate $g_{{\bf j}_\ell}$ is attributed to measurement error, perturbing $E$ to a new operator $E'$. We decompose the input state and this final measurement operator in the Pauli basis to give $\rho=\sum_P x_P P/2^n$ and $E'=\sum_P e_PP$. The average sequence fidelity is \cite{gambetta12}
\begin{equation}
F_{\mathrm seq}(\ell,E,\rho) = A_Z \alpha_Z^\ell + A_R \alpha_R^\ell + e_I
\end{equation}
where $A_Z=\sum_{P\in\ZGr\setminus\{I\}} e_Px_P$ and $A_R=\sum_{P\in\PauliGr\setminus\ZGr}e_Px_P$. In a spirit similar to simultaneous RB \cite{gambetta12}, by choosing appropriate input states, namely $|0\dots 0\rangle$ and $|+\dots +\rangle:=\sum_{b\in\{0,1\}^n}|b\rangle$, each of the two exponential decays can be observed independently. State preparation errors may lead to deviation from single exponential decay, but this is detectable. The channel parameters $\alpha_Z$ and $\alpha_R$ can be extracted by fitting the average sequence fidelity to such a decay. The corresponding depolarizing channel parameter is a weighted average $\alpha=(\alpha_Z+2^n\alpha_R)/(2^n+1)$, and the average gate error is given by $r=(2^n-1)(1-\alpha)/2^n$ (see \cite{magesan11}).

Several natural questions arise from this work. First, one might address the asymptotically optimal cost of circuit synthesis for elements of the CNOT-dihedral groups, as well as the practical question of finding optimal circuit decompositions for elements of the smallest groups. We expect optimal circuits are computationally hard to find as $n$ grows, but experimentally it is important to minimize the number of two-qubit gates. Second, unlike the Clifford group, the CNOT-dihedral group is not a $2$-design \cite{dankert09}. It would be interesting to find a group (or set) containing a non-Clifford gate and that is a $2$-design and in which benchmarking can be done efficiently. Third, our results show that we can efficiently perform RB. However, we have not addressed the precise sense in which quantum computations over the CNOT-dihedral group can be efficiently simulated. This may be a subtle problem \cite{ni13,jozsa14}. Lastly, there are generalized stabilizer formalisms, such as \cite{ni15}, and it is natural to wonder if one of these describes how this group acts on some set of states.

The remainder of the Letter is devoted to proving the various results utilized in the benchmarking procedure: canonical decomposition of $G_m$, efficient computation in $G_m$, and twirling over $G_m$, each of which is interesting in its own right. Let $m$ be general and let us briefly set some notation. The matrix representation of $G_m$ is set by identifying $g\in G_m$ to the matrix that maps $|0^n\rangle:=|00\dots 0\rangle$ to $|b\rangle:=|b_1b_2\dots b_n\rangle$ with unit phase. We define the phase-flip gates $Z|u\rangle:=(-1)^u|u\rangle$ and controlled-Z (CZ) gates $\Lambda_{12}(Z)|u,v\rangle:=(-1)^{uv}|u,v\rangle$. The support of a bit string $v\in\{0,1\}^n$ is $\supp{v}=\{j|v_j=1\}\subseteq [n]:=\{1,2,\dots,n\}$. We refer to $v$ and its support interchangeably, treating $v$ as a set and vice versa. Let $U$ be a single qubit gate and $U(v)$ denote the gate acting as $U$ only on qubits in the support of $v$. Given $J\subseteq [n]$ or elements $i,j,\dots\in [n]$, we also use the shorthand $U(J)$ and $U(i,j,\dots)$. $\PauliGr:=\langle X(j),Z(j)\rangle/\langle i\rangle$ denotes the $n$-qubit Pauli group, and we define $\XGr:=\langle X(j)|j\in [n]\rangle$, $\ZGr:=\langle Z(j)|j\in [n]\rangle$, $\CXGr:=\langle\Lambda_{ij}(X)|i,j\in [n],i\neq j\rangle$, and $\CZGr:=\langle\Lambda_{ij}(Z)|i,j\in [n],i<j\rangle$. 

{\em Canonical form of $G_m$ -- } Our first goal will be to put $G_m$ in a canonical form (the main result is contained in Theorem~\ref{thm:canonical}). The {\em rewriting identities} shown in Fig.~\ref{fig:rewrite} allow us to commute diagonal elements of $G_m$ through $\Lambda_{ij}(X)$ and $X(j)$ gates. The rules for bit-flip gates are a special case of the CNOT rules. The following Lemma follows directly from definitions and formalizes the role of the rewriting identities in understanding the group's structure.

\begin{figure}
\centering
\includegraphics[width=.47\textwidth]{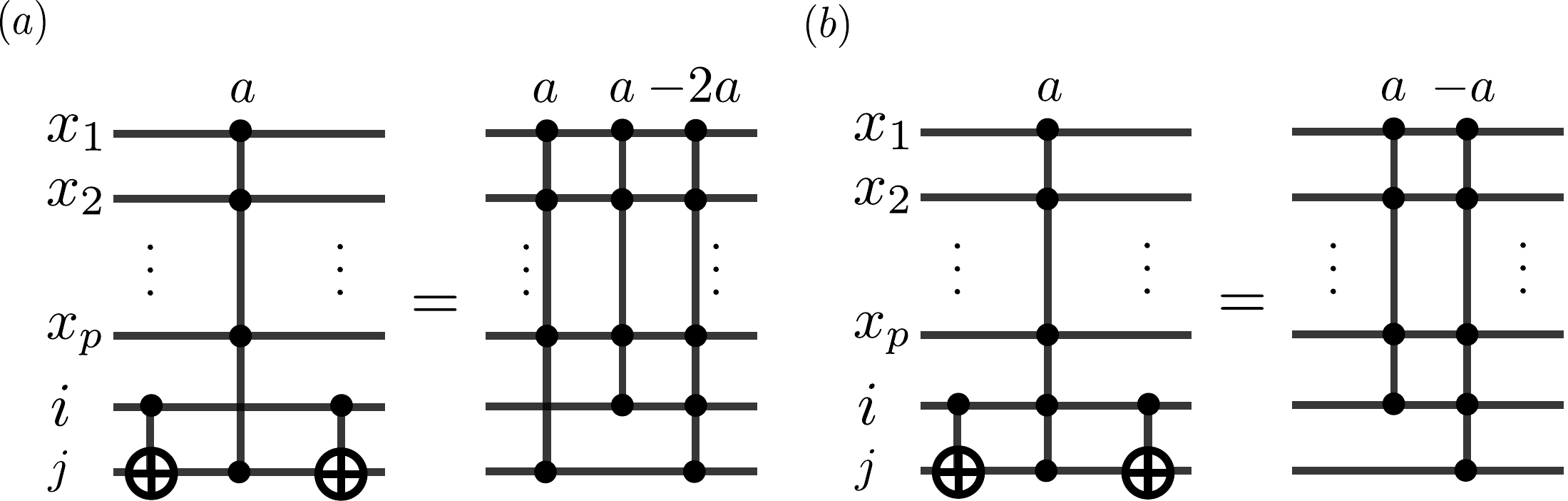}
\caption{Rewriting identities. Controlled-phase gate notation carrying the label $a$ denotes a controlled-$(Z_m)^a$ gate. (a) This is the only identity that increases the number of controls. (b) This identity preserves the number of controls.\label{fig:rewrite}}
\end{figure}

\begin{lemma}\label{lem:semidirect}
Let $W_m$ denote the subgroup of diagonal matrices of $G_m$ and let $\Pi=\langle \Lambda_{ij}(X), X(j)\rangle$ denote the subgroup of permutation matrices.
Then $G_m$ is isomorphic to a semi-direct product of groups $G_m\simeq W_m\rtimes \Pi$.
\end{lemma}

The proof of Lemma~\ref{lem:semidirect} is given in the Supplemental Material~\cite{supp}. Note that by definition, $\Pi=\XGr\rtimes \CXGr$. Since $\XGr\simeq \FF_2^n$ and $\CXGr\simeq\mathrm{GL}_n(\FF_2)$, each element of $\Pi$ can be associated to an $n$-bit string $c\in\FF_2^n$ and an $n$ by $n$ invertible $0-1$ linear transformation $B\in \mathrm{GL}_n(\FF_2)$ such that $\pi|b\rangle=|Bb\oplus c\rangle$. Here, $\FF_2$ denotes the field with two elements. Furthermore, $|\Pi|=2^n\prod_{\ell=0}^{n-1}(2^n-2^\ell)$.

It remains to better understand $W_m$ (see Lemma \ref{lemma:W} for the main result). Let $D_m$ denote the group of $2^n$ by $2^n$ diagonal unitary matrices $D$ with elements $\langle b|D|b\rangle=\omega_m^{f(b)}$. Here $f:\FF_2^n\rightarrow\ZZ_m$ is a function that assigns $m$th roots of unity to the diagonal and $\ZZ_m$ is the ring of integers modulo $m$. Since $G_m$ is generated by permutation matrices and products $m$-phase gates, $W_m\subseteq D_m$.

Let ${\cal R}\subset \ZZ_m\left[x_1,\dots,x_n\right]$ denote the polynomial ring whose elements are $p(x):=p(x_1,\dots,x_n)=\sum_{\alpha\in \ZOn} p_\alpha x^\alpha$ where $\alpha=\alpha_1\dots\alpha_n$ is a multi-index, $p_\alpha\in \ZZ_m$, and $x^\alpha=x_1^{\alpha_1}\dots x_n^{\alpha_n}$ is a monomial. The multi-index takes values in $\ZOn$ as a convenient notation since we will evaluate $p(x)$ on binary strings, so $x_j^2=x_j$. The degree of a monomial is denoted $|\alpha|$. We mainly consider ${\cal R}$ as an additive group. The next Lemma follows from the definition of group isomorphism and the fact that each function $f(b)$ can be expressed as a polynomial in ${\cal R}$.

\begin{lemma}\label{lem:iso}
Let $p(b)$ denote evaluation of $p$ on the $n$-bit binary string $b=b_1\dots b_n$ with operations in $\ZZ_m$. The function $\Phi:{\cal R}\rightarrow D_m$ given by $\langle b|\Phi(p)|b\rangle=\omega_m^{p(b)}$ is a group isomorphism. 
\end{lemma}

The proof of Lemma~\ref{lem:iso} is given in the Supplemental Material~\cite{supp}. The rewriting identities give the action of $\Pi$ on $W_m$ by conjugation. Let $\bar{W}_m:=\langle Z_m(j)\rangle$. Based on a similar application of the rewriting identities as in Lemma~\ref{lem:semidirect}, $W_m = \langle \pi\bar{W}_m\pi^\dagger\ \vert\ \pi\in\Pi\rangle$. Since $W_m\subseteq D_m\simeq {\cal R}$, $\Phi^{-1}$ associates a polynomial in ${\cal R}$ to each element of $W_m$. By our chosen convention, matrices representing elements $w\in W_m$ are given modulo a global phase factor $\langle\omega_m\rangle$ such that $w|0^n\rangle=|0^n\rangle$. Therefore the preimages $\Phi^{-1}(w)$ have zero constant term, i.e.~$p_\alpha=0$ when $|\alpha|=0$. Through $\Phi$, the rewriting identities define an action of $\Pi$ on ${\cal R}$ that respectively takes $x_1x_2\dots x_px_j$ to
\begin{equation}
-2x_1x_2\dots x_px_ix_j+x_1x_2\dots x_px_i+x_1x_2\dots x_px_j \label{eq:r1}
\end{equation}
and $x_1x_2\dots x_px_ix_j$ to
\begin{equation}
-x_1x_2\dots x_px_ix_j+x_1x_2\dots x_px_i.\label{eq:r3}
\end{equation}
Eq.~\ref{eq:r1} increments the degree of a monomial and multiplies its coefficient by $-2$, whereas Eq.~\ref{eq:r3} does not change the degree. Another way to understand iterated applications of Eq.~\ref{eq:r1} is to observe that
\begin{equation}\label{eq:xor}
y_1\oplus y_2\oplus\dots\oplus y_N = \sum_{\alpha\in\ZZ_2^N,|\alpha|\neq 0}(-2)^{|\alpha|-1}y^\alpha.
\end{equation}
This fact relates how single qubit $Z_m$ gates acting on mod $2$ linear combinations of input bits are equivalent to products of certain controlled-phase gates.

There is an element of $W_m$ corresponding to each monomial term of non-zero degree, and the coefficient of this term has the form $p_{\alpha}\in(-2)^{|\alpha|-1}\ZZ_m$, as we will now see \cite{supp}. We choose a subset of qubits $J$, fix any $j\in J$, and define a permutation gate and corresponding polynomial
\begin{equation}
\pi_{J,j} :=\prod_{\substack{k\in J\\ k\neq j}} \Lambda_{kj}(X);\ \ p_J(x) :=\sum_{\substack{\alpha\subseteq J\\|\alpha|\neq 0}}(-2)^{|\alpha|-1}x^\alpha.\label{eq:pJ}
\end{equation}
By Eq.~\ref{eq:xor}, $\Phi(p_J)=\pi_{J,j}Z_m(j)\pi_{J,j}^\dagger\in W_m$, i.e. this circuit has a polynomial with one term of degree $|J|$. Since $\Phi(Z_m(j))=x_j$, scaled monomials of successive degrees $|\alpha|$ and with coefficients in $(-2)^{|\alpha|-1}\ZZ_m$ can be generated inductively by composing these circuits. Take all linear combinations of these over $\ZZ_m$ to find

\begin{lemma}\label{lemma:W}
$W_m$ is isomorphic to the subgroup ${\cal W}<{\cal R}$ given by 
\begin{equation}
\{ p\in{\cal R}\ \vert\ p_\emptyset=0\ \mathrm{and}\ \forall\alpha\neq\emptyset,\ p_{\alpha}\in(-2)^{|\alpha|-1}\ZZ_m\ \}.
\end{equation}
\end{lemma}

We can now directly compute $|G_m|$.
\begin{corollary}
\begin{equation}
|G_m| = 2^n\prod_{\ell=0}^{n-1}(2^n-2^\ell)\prod_{t=1}^n \left(\frac{LCM(2^{t-1},m)}{2^{t-1}}\right)^{{n\choose t}}.\label{eq:order}
\end{equation}
\end{corollary}

\begin{proof}
Let $o_m(a)=\mathrm{LCM}(a,m)/a$ denote the order of $a$ in $\ZZ_m$. Observe that $(-2)^{|\alpha|-1}\ZZ_m\simeq \ZZ_{o_m(2^{|\alpha|-1})}$ as additive groups. Therefore $W_m$ is isomorphic to a direct product of additive cyclic groups $A_m:=\prod_{t=1}^n \ZZ_{o_m(2^{t-1})}^{{n\choose t}}$. This shows that $|G_m| =|A_m||\Pi|$.
\end{proof}

Putting everything together, we have

\begin{theorem}\label{thm:canonical}
Any element of $G_m$ can be written in canonical form as the composition of a sequence of phase gates (comprising an element of $W_m$ whose form is given in Lemma \ref{lemma:W}), a sequence of CNOT gates, and a sequence of bit-flip gates.
\end{theorem}

{\em Efficient computation in $G_{2^k}$ -- } Our next goal is to present efficient methods for computing with $G_m$. Suppose we consider a fixed value of $m$. Any labeling of group elements will have length proportional to $s=\log_2 |G_m|$. If $m$ is odd then $\log_2 |G_m|=(2^n-1)\log_2 m + \log_2 |\Pi|$ whereas if $m=2^k$ then $\log_2 |G_{2^k}|=\sum_{t=1}^k(k-t+1){n\choose t} + \log_2 |\Pi|$. Therefore, $s=\Omega(2^n)$ whenever $m$ is odd, and in general we cannot efficiently represent elements of $G_m$ as the number of qubits grows. However, $s=O(n^k)$ for the special case $m=2^k$, and the story is different. We focus on this special case for the remainder of this Letter.

An element $g\in G_m$ can be written as a product $g=uvw$ where $w\in W_m$ is a diagonal matrix, $v\in \CXGr$ is a CNOT circuit, and $u\in \XGr$ is a tensor product of bit-flips. This transforms $n$-qubit quantum states as $g|b\rangle = \omega_{m}^{p(b)}|Bb+c\rangle$ where $p\in {\cal W}$, $B\in \mathrm{GL}_n({\mathbb F}_2)$, and $c\in \FF_2^n$. Group elements are in bijective correspondence with the triples $(p,B,c)$. The polynomial $p$ has maximum degree $k$ and at most $\sum_{t=0}^k {n\choose t}=O(n^k)$ nonzero coefficients, each contained in $\ZZ_{2^k}$.

The product of group elements $g_1,g_2\in G_m$,
\begin{equation}
g_2g_1|b\rangle = \omega_m^{p_1(b)+p_2(B_1b+c_1)}|B_2B_1b\oplus B_2c_1\oplus c_2\rangle,
\end{equation}
is given by the triple
\begin{equation}
(p_1(x)+p_2(B_1x\oplus c_1),B_2B_1,B_2c_1\oplus c_2).
\end{equation}
The products $B_2B_1$ and $B_2c_1\oplus c_2$ can be computed in $O(n^3)$ time, and polynomials in ${\cal W}$ can be added in $O(n^k)$ ring operations.  We need to show that $p_2(B_1x\oplus c_1)$ can also be computed efficiently.

Consider a triple $(p,B,c)$ and let $B_j$ denote the $j$th row of $B$ and $J_j=\supp{B_j}$. Define $x'=Bx\oplus c$. Then for any $j\in [n]$, using Eqs.~\ref{eq:xor} and \ref{eq:pJ},
\begin{equation}
x_j'(x) =(\bigoplus\limits_{\ell\in J_j}x_j)\oplus c_j = \left\{\begin{array}{cc}
p_{J_j}(x) & \mathrm{if}\ c_j=0 \\
1-p_{J_j}(x) & \mathrm{if}\ c_j=1
\end{array}\right.
\end{equation}
has maximum degree $k$. When we substitute $x'=x_1'\dots x_n'$ into the degree $k$ polynomial $p(x)$, computations occur with coefficients in $\ZZ_{2^k}$. We compute each monomial $(x')^\alpha$ with $O(k)$ multivariate polynomial multiplications, each of which can be done term-by-term in $O(n^{2k+1})$ ring operations. We compute the term $(-2)^{|\alpha|-1}p_\alpha (x')^\alpha$ with an additional $O(n^k)$ ring operations to multiply each term of $(x')^\alpha$ by a $(-2)^{|\alpha|-1}p_\alpha$ and accumulate the result. There are $O(n^k)$ terms in $p(x)$, so the total number of ring operations to compute $p(x')$ is $O(n^{3k+1})$. If $c\neq 0^n$, it is possible that $p(x')$ has a non-zero constant term. With an additional $O(n^k)$ ring operations, $p(x')$ can be mapped to an equivalent polynomial in ${\cal W}$.

Uniformly sampling from $G_{2^k}$ is equivalent to uniformly and independently sampling from ${\cal W}$, ${\mathrm GL}_n(\FF_2)$, and $\FF_2^n$. This can be done efficiently since elements of ${\cal W}$ have maximum degree $k$; see also \cite{randall93,supp}.

Given a triple $(p,B,c)$, we synthesize a corresponding circuit from products of CNOT gates, bit-flip gates, and a single qubit $m$-phase gates. Our goal is to efficiently synthesize a circuit whose size (number of gates) is polynomial in $n$ but not to optimize this circuit. We independently synthesize circuits coinciding with $p$, $B$, and $c$. Since $c$ corresponds to $X(c)$, and a CNOT circuit for $B$ can be found by Gaussian elimination \cite{gottesman97}, the new part of the algorithm synthesizes a circuit for $p$.

We describe the circuit synthesis for $p$ informally. The algorithm proceeds in $k$ rounds. Begin by initializing a working polynomial $q(x)\leftarrow p(x)$, set a round counter $t\leftarrow k$, and set a quantum circuit $U\leftarrow I$. Here ``$\leftarrow$'' denotes assignment. In round $t$, we synthesize a circuit corresponding to a polynomial $p^{(t)}(x)$ that coincides with $q(x)$ on its degree $t$ terms. For each of the $O(n^t)$ degree-$t$ terms $(-2)^{|\alpha|-1}p_\alpha x^\alpha$ of $q(x)$, we apply the constant-sized circuit $g_\alpha:=\pi_{J,j}\left(Z_{2^k}(j)\right)^{p_\alpha}\pi_{J,j}^\dagger$ setting $U\leftarrow g_\alpha U$, where $J=\supp{\alpha}$ as in the proof of Lemma~\ref{lemma:W}. The product of the $g_\alpha$ corresponds to $p^{(t)}(x):=\prod_{\alpha\subseteq [n],|\alpha|=t}p_\alpha p_J(x)$. Therefore we update $q(x)\leftarrow q(x)-p^{(t)}(x)$, which now has maximum degree $t-1$, decrement the round counter, and proceed to the next round. The algorithm terminates when $q(x)=0$ and $t=0$. The total algorithm run-time and circuit size of the output $U$ is $O(n^k)$.

{\em Twirling over $G_{2^k}$ --} A quantum channel is a completely-positive trace-preserving map whose operator sum decomposition is ${\cal E}(\rho) = \sum_k A_k\rho A_k^\dagger$
where $\sum_k A_k^\dagger A_k=I$. The twirl of ${\cal E}$ over a finite group $G$ ($G$-twirl) is given by
\begin{equation}\label{eq:twirl}
\bar{\cal E}_G(\rho) := \frac{1}{|G|}\sum_{U\in G}U^\dagger {\cal E}(U\rho U^\dagger)U.
\end{equation}

In what follows, we use several facts about group twirls. If $G=AB$ is a direct product of groups then $\bar{{\cal E}}_G(\rho)=\overline{({\bar {\cal E}}_A)}_B(\rho)$, and if $A$ is a normal subgroup of $G$ (denoted $A\triangleleft G$) then $\bar{{\cal E}}_G(\rho)=\overline{({\bar {\cal E}}_A)}_{G/A}(\rho)$, where the twirl over the factor group $G/A$ is over a set of coset representatives. Twirling any map over the Pauli group produces a Pauli channel \cite{dankert09}. Consider a Pauli channel ${\cal E}(\rho) = \sum_{Q\in \PauliGr} \eta_Q Q\rho Q$. Twirl this channel over any finite group $G$ that has a permutation action on the set $\PauliGr$. The orbit of $P\in \PauliGr$ is $O_P:=\{V^\dagger PV\ |\ V\in G\}$ and the stabilizer is $S_P:=\{ V\in G\ |\ V^\dagger PV=P\}$. The orbits define an equivalence relation $P\sim Q$ if and only if $O_P=O_Q$. This relation partitions $\PauliGr$ into a disjoint union of orbits. By the orbit-stabilizer theorem and Lagrange's theorem \cite{artin98}, $|O_P|=|G/S_P|=|G|/|S_P|$. Therefore the twirl, Eq.~\ref{eq:twirl}, can be written
\begin{equation}
\bar{\cal E}_{G}(\rho) = \sum_{C\in {\cal C} } \sum_{P\in O_C}\left(\frac{\sum_{Q\in O_C}\eta_Q}{|O_C|}\right) P\rho P,\label{eq:orbits}
\end{equation}
where ${\cal C}$ is a set of representative elements, one from each orbit.

These facts allow us to compute the twirl over $G_{2^k}$ when $k>1$ by expressing it as a sequence of twirls. We begin by decomposing the group. Let $\tilde{W}_{2^k}:=W_{2^k}\setminus (\bar{W}_{2^k}\setminus\{I\})$ and recall that $\bar{W}_{2^k}:=\langle Z_{2^k}(j)\rangle$. Then $W_{2^k}=\tilde{W}_{2^k}\bar{W}_{2^k}$. Since $\CZGr\triangleleft\tilde{W}_{2^k}$ and $\ZGr\triangleleft\bar{W}_{2^k}$, we form the corresponding factor groups. Therefore an element $w\in W_{2^k}$ can be written as $w=\tilde{w}\bar{w}=\tilde{w}_1\tilde{w}_2\bar{w}_1\bar{w}_2$ where $\tilde{w}_1$ labels cosets $\tilde{w}_1\CZGr$, $\tilde{w}_2\in\CZGr$, $\bar{w}_1$ labels cosets $\bar{w}_1\ZGr$, and $\bar{w}_2\in\ZGr$. Finally, by Lemma~\ref{lem:semidirect}, any element $g\in G_{2^k}$ factors as $g=uvw$ where $u\in\XGr$, $v\in\CXGr$, and $w\in W_{2^k}$. Therefore, we have $g=u\bar{w}_2'v\tilde{w}_2\bar{w}_1\tilde{w}_1$ where $\bar{w}_2'=v\bar{w}_2v^\dagger\in\ZGr$.

Our strategy is to use the decomposition to express the $G_{2^k}$-twirl as a sequential $\PauliGr$-twirl, $\CXGr$-twirl, $\CZGr$-twirl, $\bar{W}_{2^k}/\ZGr$-twirl, and $\tilde{W}_{2^k}/\CZGr$-twirl. Each twirl can be computed in a straightforward manner using the facts we have described and reduces the number of independent parameters describing the channel until we have twirled over the whole of $G_{2^k}$ \cite{supp}. The final twirled map is
\begin{equation}
{\cal E}(\rho) = \beta_I\rho + \beta_Z\sum_{P\in \ZGr\setminus\{I\}} P\rho P + \beta_R \sum_{P\in\PauliGr\setminus\ZGr} P\rho P.
\end{equation}
In the Liouville representation in the Pauli basis which has matrix elements $R^{({\cal E})}_{PQ}=\mathrm{Tr}(P{\cal E}(Q))/4^n$ where $P$ and $Q$ are $n$-qubit Pauli operators, this map has three diagonal blocks corresponding to $I$, $\ZGr\setminus\{I\}$, and $\PauliGr\setminus\ZGr$ with elements $1$, $\alpha_Z:=1-4^n\beta_R$, and $\alpha_R:=1-2^n\beta_Z-(4^n-2^n)\beta_R$, respectively.

{\em Conclusion -} Our results enable scalable benchmarking of a natural family of non-Clifford circuits related to quantum error-correcting codes. Because our procedure scales, in principle it allows efficient benchmarking of isolated non-Clifford gates as well as large sub-circuits for state distillation \cite{bravyi05,duclos15} or repeat-until-success protocols \cite{paetznick14}. These sub-circuits can be characterized with our procedure using physical gates or logical gates on protected qubits. Together with standard Clifford benchmarking, our procedures enable characterization of the full range of gates used in the leading fault-tolerant quantum computing protocols. As multi-qubit benchmarking is well within experimental reach, we expect an optimized implementation of our procedure to be quite practical.

\begin{acknowledgments}
We acknowledge support from ARO under contract W911NF-14-1-0124.
\end{acknowledgments}

\end{thebibliography}

\newpage
\pagebreak
\widetext
\newpage
\begin{center}
\textbf{\large Supplemental Materials: Scalable randomized benchmarking of non-Clifford gates}
\end{center}
\vspace{1cm}
\setcounter{equation}{0}
\setcounter{figure}{0}
\setcounter{table}{0}
\setcounter{page}{1}
\makeatletter
\renewcommand{\theequation}{S\arabic{equation}}
\renewcommand{\thefigure}{S\arabic{figure}}
\renewcommand{\bibnumfmt}[1]{[S#1]}
\renewcommand{\citenumfont}[1]{S#1}

Let $\Lambda_{x_1x_2\dots x_pj}(U)$ denote a controlled-$U$ gate targeting qubit $j$ with controls $x_1$, $x_2$, $\dots$, $x_p$. Let ${\cal O}^m$ and ${\cal E}^m$ denote odd and even strings of length $m$ respectively. Given a subset $J\subseteq [n]$, $\bar{J}:=[n]\setminus J$. For a support $J$ and string $u\in\{0,1\}^{|J|}$, $u[J]\in\ZZtn$ denotes the string whose restriction to $J$ is $u$ and is zero otherwise. The indicator function $\delta(u,v)$ is one if $u=v$ and zero otherwise. 

The circuit diagrams in Fig.~\ref{fig:rewrite} are an intuitive way to present the rewriting identities, but they can also be written concisely as
\begin{align}
\Lambda_{ij}(X)\Lambda_{x_1x_2\dots x_pj}(Z_{m}^a)\Lambda_{ij}(X) & =\Lambda_{ix_1x_2\dots x_pj}(Z_{m}^{-2a})\Lambda_{x_1x_2\dots x_pi}(Z_{m}^a)\Lambda_{x_1x_2\dots x_pj}(Z_{m}^a), \label{eq:c1}\\
\Lambda_{ij}(X)\Lambda_{x_1x_2\dots x_pi}(Z_{m}^a)\Lambda_{ij}(X) & =\Lambda_{x_1x_2\dots x_pi}(Z_{m}^a), \label{eq:c2}\\
\Lambda_{ij}(X)\Lambda_{x_1x_2\dots x_pij}(Z_{m}^a)\Lambda_{ij}(X) & =\Lambda_{x_1x_2\dots x_pij}(Z_{m}^{-a})\Lambda_{x_1x_2\dots x_pi}(Z_{m}^a). \label{eq:c3}
\end{align}

For completeness, we give the proofs of Lemma~\ref{lem:semidirect}, Lemma~\ref{lem:iso}, Equation~\ref{eq:xor}, and Lemma~\ref{lemma:W}.

{\em Proof of Lemma~\ref{lem:semidirect}.} By definition of the semi-direct product, we must show that (a) $W_m\cap\Pi=\{I\}$, where $I$ denotes the identity matrix, (b) $G_m=\Pi W_m$, and (c) $W_m\triangleleft G_m$. Every element of $\Pi$ is a permutation matrix. Therefore the only diagonal matrix in $\Pi$ is the identity matrix, proving (a). Consider an element $g=g_N\dots g_2g_1\in G_m$ expressed as a product of generators. Suppose the $\ell$th generator is $g_\ell=Z_m(j)\in W_m$. By repeated application of rewriting identities, we can write $g=g_m\dots g_{\ell+1}g_{\ell-1}\dots g_2g_1w_\ell$ where $w_\ell=g_1^\dagger\dots g_{\ell-1}^\dagger g_\ell g_{\ell-1}\dots g_1\in W_m$. Repeating this argument for all occurrences of $Z_m(j)$ in $g$, we obtain $g=\pi w$ where $\pi\in\Pi$ and $w\in W_m$; thus $G_m<\Pi W_m$. Conversely, since $\Pi$ and $W_m$ are both subgroups of $G_m$, any $\pi w\in \Pi W_m<G_m$. This proves (b). By (b), each element of $G_m$ can be written as $g=\pi w_g$ for some $\pi\in\Pi$ and $w_g\in W_m$, so $gwg^\dagger=\pi w_gww_g^\dagger \pi^\dagger=\pi w'\pi^\dagger$ where $w'=w_gww_g^\dagger\in W_m$. Decomposing $\pi$ into a product of generators $X(j)$ and $\Lambda_{ij}(X)$ gates and applying rewriting identities to each element in the product, we find $\pi w'\pi^\dagger=w''\in W_m$ and thus prove (c). \qed

{\em Proof of Lemma~\ref{lem:iso}.} We need to show that $\Phi(p+p')=\Phi(p)\Phi(p')$, $\Phi(-p)=\Phi(p)^\dagger$, and that $\Phi$ is a bijection. First, $(p+p')(x)=p(x)+p'(x)$ implies that $\langle b|\Phi(p+p')|b\rangle=\omega_m^{p(b)}\omega_m^{p'(b)}$. On the other hand, $\langle b|\Phi(p)\Phi(p')|b\rangle =\omega_m^{p(b)}\omega_m^{p'(b)}$. Next, $\langle b|\Phi(p)^\dagger|b\rangle=(\omega_m^\ast)^{p(b)}=(\omega_m^{-1})^{p(b)}=\omega_m^{-p(b)}=\langle b|\Phi(-p)|b\rangle$.
Suppose that $\Phi(p)=\Phi(p')$ for some $p,p'\in{\cal R}$. For all $n$-bit strings $b$, we have $\omega_m^{p(b)}=\omega_m^{p'(b)}$ and therefore $p=p'$. Thus $\Phi$ is injective. Now consider some $D\in D_m$ with matrix elements $\langle b|D|b\rangle=\omega_m^{f(b)}$. We write $f(x)=\sum_{s\in\FF_2^n} f(s)\delta(x,s)$ and observe
\begin{equation}
\delta(x,s)=\prod_{j=1}^n\left\{\begin{array}{cc}
x_j & \mathrm{if}\ s_j=1 \\
1-x_j & \mathrm{if}\ s_j=0 \end{array}\right\}\in {\cal R}.
\end{equation}
Therefore, $f(x)\in {\cal R}$ and $D(f)=D$, proving that $\Phi$ is surjective. \qed

{\em Proof of Equation~\ref{eq:xor}.} It can be proven by induction on $N$. It is obviously true when $N=1$. If we assume that Eq.~\ref{eq:xor} holds for some $N$, then
\begin{equation}
(\bigoplus_{j=1}^N x_j)\oplus x_N = [x_1\oplus\dots\oplus x_N](1-x_{N+1}) + (1-[x_1\oplus\dots\oplus x_N])x_{N+1}
\end{equation}
implies that it holds for $N+1$ and hence for all $N\geq 1$. \qed

{\em Proof of Lemma~\ref{lemma:W}.} Given a non-empty subset $J\subseteq [n]$, fix any element $j\in J$, and let 
\begin{equation}
\pi_{J,j} :=\prod_{\substack{k\in J\\ k\neq j}} \Lambda_{kj}(X)\ \mathrm{and}\ \ p_J(x) :=\sum_{\substack{\alpha\subseteq J\\|\alpha|\neq 0}}(-2)^{|\alpha|-1}x^\alpha.
\end{equation}
By Eq.~\ref{eq:xor}, we see that $\Phi(p_J)=\pi_{J,j}Z_m(j)\pi_{J,j}^\dagger\in W_m$. We begin by observing that $\Phi(Z_m(j))=x_j$. Now suppose that for any $\alpha$ satisfying $0<|\alpha|\leq\ell$, there exists a $w_\alpha\in W_m$ such that $\Phi^{-1}(w_\alpha)=(-2)^{|\alpha|-1}x^\alpha$. We will show that for any $\alpha$ such that $|\alpha|=\ell+1\leq n$, there exists $w_{\alpha}\in W_m$ such that $\Phi^{-1}(w_{\alpha})=(-2)^{|\alpha|-1}x^\alpha$. Let $J=\supp{\alpha}$ and
\begin{equation}
w_\alpha = \pi_{J,j}Z_m(j)\pi_{J,j}^\dagger\prod_{\substack{\beta\subset J\\|\beta|\neq 0}} w_\beta^\dagger.
\end{equation}
Then
\begin{equation}
\Phi^{-1}(w_\alpha) = \Phi^{-1}(\pi_{J,j}Z_m(j)\pi_{J,j}^\dagger) + \sum_{\substack{\beta\subset J\\|\beta|\neq 0}} \Phi^{-1}(w_\beta^\dagger)
= p_J(x) - \sum_{\substack{\beta\subset J\\|\beta|\neq 0}} (-2)^{|\beta|-1}x^\beta = (-2)^{|\alpha|-1}x^\alpha.
\end{equation}
By induction, there exists a $w_\alpha\in W_m$ such that $\Phi^{-1}(w_\alpha)=(-2)^{|\alpha|-1}x^\alpha$ for any $\alpha$ satisfying $0<|\alpha|\leq n$. Eq.~\ref{eq:r1} and Eq.~\ref{eq:r3} are the complete set of relations for conjugating $W_m$ by $\Pi$, and therefore every monomial $x^\alpha$ in $\Phi^{-1}(W^{(n)}_m)$ has a coefficient that is a multiple of $2^{|\alpha|-1}$. We take all linear combinations of these monomials over $\ZZ_m$. \qed

The corollary to Lemma~\ref{lemma:W} allows us to compute the order of $G_m$ for any $m$ and any number of qubits $n$. Both to verify this formula and get a sense of the size of the group, we have tabulated orders of the smallest groups $G_m$ in Table~\ref{tab:orders} by explicitly constructing them. The group order increases rapidly with $n$ but is comparable in size to the Clifford group.

\begingroup\squeezetable
\begin{table}[h]
\centering
\begin{ruledtabular}\begin{tabular}{rrrrrrrrr}
$m$ & 1 & 2 & 3 & 4 & 5 & 6 & 7 &  8 \\ \hline
$n=1$ & 2 & 4 & 6 & 8 & 10 & 12 & 14 & 16\footnote{smallest group for benchmarking $T$ gate.}\\
$n=2$ & 24 & 96 & 648 & 768 & 3000 & 2592 & 8232 & 6144\footnote{smallest group for benchmarking  $\Lambda_{ij}(\sqrt{Z})$ gate.}\\
$n=3$ & 1344 & 10752 & 2939328 & 688128 & 105000000 & 23514624 & 1106841792 &  88080384\footnote{smallest group for benchmarking $\Lambda_{ijk}(Z)$ gate.}
\end{tabular}\end{ruledtabular}
\caption{Group orders for small non-Clifford matrix groups with $n\leq 3$. All groups were explicitly constructed from generating sets.\label{tab:orders}}
\end{table}
\endgroup

We provide more details about how to sample a uniformly random element of $G_m$. Assume we have access to a string of uniformly random bits. We take $n$ bits as the element $c\in\ZZ_2^n$. In time $M(n)+O(n^2)$ and with $n^2+3$ random bits, where $M(n)$ is the time to multiply two $n$ by $n$ matrices, we can generate a uniformly random non-singular matrix $B\in\mathrm{GL}_n(\FF_2)$ \cite{S_randall93}. Finally, for all $\alpha$ satisfying $0<|\alpha|\leq k$, we take $k$ random bits to draw a uniformly random element from $\ZZ_{2^k}$, which we multiply by $(-2)^{|\alpha|-1}$ modulo $2^k$ to obtain the corresponding coefficient of the polynomial $p\in {\cal W}$. There are $O(n^k)$ non-zero coefficients, so we consume $O(kn^k)$ random bits to sample an element uniformly at random from ${\cal W}$.

We use several facts to compute the $G_m$-twirl. We prove these facts now. If $G=AB$ is a direct product of groups $A$ and $B$, then
\begin{equation}
\bar{{\cal E}}_G(\rho) = \frac{1}{|G|}\sum_{U\in G} U^\dagger {\cal E}(U\rho U^\dagger)U = \frac{1}{|A||B|}\sum_{S\in A,T\in B} T^\dagger S^\dagger {\cal E}(ST\rho T^\dagger S^\dagger)ST = \frac{1}{|B|}\sum_{T\in B} T^\dagger \bar{{\cal E}}_A(T\rho T^\dagger)T = \overline{({\bar {\cal E}}_A)}_B(\rho).
\end{equation}
If instead $A\triangleleft G$, then
\begin{equation}
\bar{{\cal E}}_G(\rho) = \frac{1}{|G/A|}\sum_{T\in G/A} T^\dagger \bar{{\cal E}}_A(T\rho T^\dagger)T = \overline{({\bar {\cal E}}_A)}_{G/A}(\rho)
\end{equation}
where the sum over $T\in G/A$ is over a set of coset representatives, i.e.~each $T$ is any representative in the coset $TA$. It is well-known that twirling any map over the Pauli group produces a Pauli channel. This can be seen from
\begin{align}
\bar{\cal E}_{\PauliGr} & = \sum_k \frac{1}{|\PauliGr|} \sum_{P\in \PauliGr} P^\dagger A_kP\rho P^\dagger A_k^\dagger P  = \sum_k \frac{1}{|\PauliGr|} \sum_{P,Q,R\in \PauliGr} \gamma_Q^{(k)}\gamma_R^{(k)\ast} PQP\rho PRP = \sum_{Q\in \PauliGr} \left(\sum_k |\gamma_Q^{(k)}|^2\right) Q\rho Q.
\end{align}
We have expressed each Kraus operator in the Pauli operator basis $A_k=\sum_Q \gamma_Q^{(k)} Q$. We wrote $PQP=(-1)^{\omega(P,Q)}Q$ where $\omega(P,Q)=0$ if $[P,Q]=0$ and $\omega(P,Q)=1$ otherwise, and used the fact that $\frac{1}{|\PauliGr|}\sum_{P\in \PauliGr}(-1)^{\omega(P,Q)+\omega(P,R)}=\delta(P,Q)$. Trace preservation ensures that $\sum_Q\sum_k |\gamma_Q^{(k)}|^2=1$. This twirl reduces a general map to one described by $4^n$ parameters. For convenience let $\eta_Q:=\sum_k |\gamma_Q^{(k)}|^2$.

Eq.~\ref{eq:orbits} can be derived from the definition of the twirl as follows
\begin{equation}
\bar{\cal E}_{G}(\rho) = \sum_{Q\in \PauliGr} \sum_{P\in O_Q}\frac{1}{|O_Q||S_Q|}|S_Q|\eta_Q P\rho P = \sum_{C\in {\cal C} } \sum_{Q\in O_C}\sum_{P\in O_Q}\frac{1}{|O_Q||S_Q|}|S_Q|\eta_Q P\rho P = \sum_{C\in {\cal C} } \sum_{P\in O_C}\left(\frac{\sum_{Q\in O_C}\eta_Q}{|O_C|}\right) P\rho P.
\end{equation}

For completeness, we conclude now with details about each step of the $G_{2^k}$-twirl: the $\CXGr$-twirl, the $\CZGr$-twirl, the explicit orbits for each, the $\bar{W}_{2^k}/\ZGr$-twirl, and the final $\tilde{W}_{2^k}/\CZGr$-twirl.

The $\CXGr$-twirl is given by Eq.~\ref{eq:orbits}, and we need only enumerate the orbits. There are $5$ orbits corresponding to the representative elements ${\cal C}_\CXGr:=\{I,X(1),Z(1),Y(1),Y(1,2)\}$. Simple counting gives $|O_{X(1)}|=|O_{Z(1)}|=2^n-1$, $|O_{Y(1)}|=2^{n-1}(2^n-1)$, and $|O_{Y(1,2)}|=4^n/2-3\cdot 2^{n-1}+1$. Note that $|O_{Y(1,2)}|=|O_{X(1)}|(2^{n-1}-1)$ and $|O_{X(1)}|+|O_{Y(1,2)}|=|O_{Y(1)}|$. After the $\CXGr$-twirl, the Pauli channel is described by $5$ parameters $\beta_I$, $\beta_X$, $\beta_Z$, $\beta_Y$, and $\beta_{YY}$ that correspond to averages over the Pauli channel parameters in each respective orbit.

The $\CZGr$-twirl has an exponential number of orbits. Let $u^\ast\in [n]$ denote the first non-zero coordinate of a non-zero $u\in\FF_2^n$. The orbits correspond to representative elements ${\cal C}_\CZGr:=\XGr\cup (\ZGr\setminus\{I\})\cup\{X(u)Z(u^\ast)\ :\ u\in\FF_2^n,u\neq 0\}$. There are $2^n-1$ orbits of each non-identity type and $|O_{X(u)}|=|O_{X(u)Z(u^\ast)}|=2^{n-1}$ for all $u$. Comparing the $\CZGr$ orbits to the $\CXGr$ orbits, we see that $O_{Z(1)}=\cup_u O_{Z(u)}$, $O_{Y(1)}=\cup_u O_{X(u)Z(u^\ast)}$, and $O_{X(1)}\cup O_{Y(1,2)}=\cup_u O_{X(u)}$. Furthermore $O_{X(u)}$ contains one element from $O_{X(1)}$ and $2^{n-1}-1$ elements from $O_{Y(1,2)}$. Therefore parameters for $O_{X(1)}$ and $O_{Y(1,2)}$ are averaged into a new parameter $\beta_{X-YY}=(\beta_X+(2^{n-1}-1)\beta_{YY})/2^{n-1}$ that is simply the average over all parameters $\eta_Q$ with $Q\in O_{X(1)}\cup O_{Y(1,2)}$ upon substituting $\beta_X$ and $\beta_{YY}$. After the $\CZGr$-twirl, the Pauli channel is described by $4$ parameters.

These are the explicit orbits of the action of $\CXGr$ and $\CZGr$ on $\PauliGr$. The $\CXGr$ orbits are
\begin{align*}
O_{I}  & = \{I\} \\
O_{X(1)} & = \XGr\setminus\{I\} \\
O_{Z(1)} & = \ZGr\setminus\{I\} \\
O_{Y(1)} & = \{ X(J)Z(u[J]+v[\bar{J}])\ :\ |J|>0,u\in{\cal O}^{|J|} \} \\
O_{Y(1,2)} & = \{ X(J)Z(u[J]+v[\bar{J}])\ :\ |J|>1,u\in{\cal E}^{|J|} \} \cup \{ X(J)Z(v[\bar{J}])\ :\ |J|>0,|v|>0 \}
\end{align*}
where $J\subseteq [n]$ ranges over all subsets and $v\in\FF_2^{|\bar{J}|}$ ranges over all strings. It is straightforward to verify that these are the orbits since $\Lambda(X)$ preserves the parity of the number of $Y$-type tensor factors and does not mix $X$-type and $Z$-type factors. The $\CZGr$ orbits are
\begin{align*}
O_{I} & = \{I\} \\
O_{Z(u)} & = \{Z(u)\} \\
O_{X(u)} & = \{X(J)Z(w[J]+v[\bar{J}])\ :\ w\in{\cal E}^{|J|}\} \\
O_{X(u)Z(u^\ast)} & = \{X(J)Z(w[J]+v[\bar{J}])\ :\ w\in {\cal O}^{|J|}\}
\end{align*}
where $u\in\FF_2^n$, $u\neq 0$, $v\in\FF_2^{|\bar{J}|}$ and $J=\supp{u}$. This is also straightforward to show from the action of $\Lambda(Z)$ on $\PauliGr$.

The $\bar{W}_{2^k}/\ZGr$-twirl is over coset representatives of $\langle Z_{2^k}(j)\rangle/\ZGr$.  Let $P=X(u)Z(v)\in\PauliGr$ and $P'(w)=X(J)Z(w[J]+v)$ where $J=\supp{u}$. Twirling all tensor factors gives
\begin{equation}
P\rho P \mapsto \frac{1}{2^{|J|}}\sum_{w\in \FF_2^{|J|}} P'(w)\rho P'(w)^\dagger. 
\end{equation}
Tensor factors that are $X$ or $Y$ become equal sums of both, whereas $Z$ commutes. Therefore the $\CZGr$-twirl classes $O_{X(u)}$ and $O_{X(u)Z(u^\ast)}$ for fixed $u$ (and fixed $J=\supp{u}$) are averaged together equally. Since these orbits all have the same size and partition $O_{X(1)}\cup O_{Y(1,2)}$ and $O_{Y(1)}$ respectively, the parameters $\beta_{X-YY}$ and $\beta_Y$ are averaged into a new parameter $\beta_R:=(\beta_{X-YY}+\beta_Y)/2$. Again this is simply the average of all parameters $\eta_Q$ with $Q\in O_{X(1)}\cup O_{Y(1,2)}\cup O_{Y(1)}$. After the $\bar{W}_{2^k}/\ZGr$-twirl, the Pauli channel is described by $3$ parameters.

The detailed calculations for the $\bar{W}_{2^k}/\ZGr$-twirl are as follows. Conjugating a single qubit Pauli operator $X^uZ^v$ by a coset representative $(Z_{2^k})^a$, such as occurs in a term of Eq.~\ref{eq:twirl}, gives
\begin{equation}
\left[\cos(2\pi a/2^k)I+i\sin(2\pi a/2^k)Z\right]^uX^uZ^v.
\end{equation}
In the $(u,v)=(1,0)$ case, when we twirl only the first tensor factor, $P\rho P$ is taken to
\begin{align}
\frac{1}{2^{k-1}}\sum_{a=0}^{2^{k-1}-1}P_{2:n}\left[\right. & \cos^2(2\pi a/2^k)X\rho X- \cos(2\pi a/2^k)\sin(2\pi a/2^k)Y\rho X \\
& -\cos(2\pi a/2^k)\sin(2\pi a/2^k)X\rho Y+ \left.\sin^2(2\pi a/2^k)Y\rho Y\right]P_{2:n}
\end{align}
where $P_{2:n}$ denotes remaining tensor factors of $P$ on qubits $\{2,\dots,n\}$. Applying the identity $\cos^2 u+\cos^2(\pi/2+u)=1$, we can write
\begin{equation}
\sum_{a=0}^{2^{k-1}-1}\frac{1}{2^{k-1}}\cos^2(2\pi a/2^k) = \sum_{a=0}^{2^{k-2}-1}\frac{1}{2^{k-1}}\left(\cos^2(2\pi a/2^k)+\cos^2(\pi/2+2\pi a/2^k)\right) = 1/2.
\end{equation}
Similar arguments allow us to sum the remaining terms. The odd functions vanish and the even functions average to $1/2$.

The final steps of the twirl over $G_m$ are straightforward but we provide more details. The diagonal PTM representation of ${\cal E}(\rho)$ is
\begin{equation}
R^{(\cal E)}_{QQ}=\beta_I+\beta_Z\sum_{P\in \ZGr\setminus\{I\}}(-1)^{\omega(P,Q)}+\beta_R\sum_{P\in\PauliGr\setminus\ZGr}(-1)^{\omega(P,Q)}.
\end{equation}
Any non-identity Pauli operator commutes with exactly half of the Pauli group and anti-commutes with the other half. If $Q\in \ZGr\setminus\{I\}$ then $R^{\cal E}_{QQ}=\beta_I+(2^n-1)\beta_Z-2^n\beta_R$. Likewise, $R^{(\cal E)}_{II}=\beta_I+(2^n-1)\beta_Z+(4^n-2^n)\beta_R=1$. When $Q\in\PauliGr\setminus\ZGr$, $R^{(\cal E)}_{QQ}=\beta_I-\beta_Z$. Therefore the map is diagonal and has three blocks proportional to the identity. Consider a unitary operator $U$. If $U$ commutes with $Q$ or $Q'$, then by the cyclic property of the trace $R^{(U)}_{QQ'}=\delta(Q,Q')$. If $U$ is a diagonal gate then $R^{(U)}$ is block diagonal with identity matrices in the $\ZGr$ blocks, all off diagonal blocks equal to zero, and an arbitrary $\PauliGr\setminus\ZGr$ block. Since $R^{(U^\dagger)}=(R^{(U)})^{-1}$ and every element of $\tilde{W}_{2^k}/\CZGr$ is diagonal, the final twirled map is
\begin{equation}
\frac{1}{|G|}\sum_{U\in \tilde{W}_{2^k}/\CZGr} R^{(U^\dagger)}R^{(\cal E)}R^{(U)} = R^{(\cal E)}.
\end{equation}

\end{document}